\documentclass[12pt,leqno,oneside]{amsart}

\usepackage{amsmath}
\usepackage{amsaddr}
\usepackage[authoryear,round]{natbib}
\usepackage{amsfonts}
\usepackage{amssymb}
\usepackage{times}
\usepackage{multirow}
\usepackage{color}
\usepackage[normalem]{ulem}
\usepackage{graphicx}
\usepackage{epstopdf}

\newtheorem{theorem}{Theorem}
\newtheorem{lemma}[theorem]{Lemma}
\newtheorem{conjecture}[theorem]{Conjecture}
\newtheorem{cor}[theorem]{Corollary}

\theoremstyle{definition}
\newtheorem{definition}[theorem]{Definition}
\newtheorem{remark}{Remark}

\def\:={\mathrel{\mathop:}=}

\newcommand{\commentout}[1]{}

\def\g{\gamma}

\def\d{{\delta}}

\def\D{\Delta}

\def\and{\wedge}

\def\and{\wedge}
\def\d{\delta}

\def\R{\mathbb{R}}

\def\hide#1{}

\let\oldmarginpar\marginpar
\renewcommand\marginpar[1]{\-\oldmarginpar[\raggedleft\footnotesize #1]%
{\raggedright\footnotesize #1}}

\def\head{\mathsf h}
\def\tail{\mathsf t}

\DeclareMathOperator{\Succ}{succ}

\newcommand \abs[1] {|{#1}|}

\def \ve{\varepsilon}
\newcommand \event[1]{\{{#1}\}}

\begin{document}
\title{How to Gamble Against All Odds}
\author{Gilad Bavly}
\address{
Department of Statistics and Operations Research\\
School of Mathematical Sciences\\
Tel Aviv University
}
\email{gilbav@math.huji.ac.il}
\thanks{
Gilad Bavly acknowledges ISF grants 538/11 and 323/13, and BSF grant 2010253.
}
\author{Ron Peretz}
\address{Department of Mathematics\\London School of Economics and Political Science}
\email{ronprtz@gmail.com}
\begin{abstract}
A decision maker observes the evolving state of the world while constantly trying to predict the next state given the history of past states. The ability to benefit from such predictions depends not only on the ability to recognize patters in history, but also on the range of actions available to the decision maker.

We assume there are two possible states of the world. The decision maker is a gambler who has to bet a certain amount of money on the bits of an announced binary sequence of states. If he makes a correct prediction he wins his wager, otherwise he loses it.

We compare the power of betting strategies (aka \emph{martingales}) whose wagers take values in different sets of reals. A martingale whose wagers take values in a set $A$ is called an $A$-martingale. A set of reals $B$ \emph{anticipates} a set $A$, if for every $A$-martingale there is a countable set of $B$-martingales, such that on every binary sequence on which the $A$-martingale gains an infinite amount at least one of the $B$-martingales gains an infinite amount, too.

We show that for two important classes of pairs of sets $A$ and $B$, $B$ anticipates $A$ if and only if the closure of $B$ contains $rA$, for some positive $r$. One class is when $A$ is bounded and $B$ is bounded away from zero; the other class is when $B$ is well ordered (has no left-accumulation points). Our results generalize several recent results in algorithmic randomness and answer a question posed by \cite{chalcraft12}.\end{abstract}
\maketitle
\noindent \textbf{Keywords:} repeated games, gambling, algorithmic randomness, pseudo-randomness, predictability.\\
\textbf{JEL classification:} C72, C73.

\section{Introduction}
Randomness and computation are related to two basic elements of boun\-ded rationality modeling: (a) players cannot generate truly random actions; (b) they cannot implement arbitrarily complex strategies. When rationality is computationally constrained, true randomness can be replaced by pseudo-randomness: one need not generate truly random actions, just actions whose pattern is too complicated for any rationally bounded agent to be able to recognize.

Algorithmic randomness is the field of computer science that studies the relations between computability and randomness. It provides the rigorous definitions and analytic tools to address questions involving the tension between the two above elements of bounded rationality. On the other hand, the present paper applies game theory to solve a problem from algorithmic randomness by recasting the problem as an extensive form game.

\subsection{Gambling house game}

Gambler 0, the cousin of the casino owner, repeatedly bets on red/black. She is only allowed to wager even (positive) integers, i.e., she can bet 4 dollars on black, then 10 on red, then 2 on black, etc. (``betting 4 on black'' means that you gain 4 if black occurred, and lose 4 if red did). Then the \emph{regular} gamblers make their bets, and they are only allowed to wager odd integers. The red/black outcome is not decided by a lottery; rather, the casino owner chooses red or black. He attempts to aid gambler 0 by his choices.\footnote{Following the standard algorithmic randomness formalism, it is assumed that the gamblers do not observe each other's bets or gains, but only observe the casino sequence of reds and blacks. This assumption turns out to be immaterial for our results.}

All gamblers start with some finite, not necessary equal, wealth. Going into debt is not allowed, i.e., you cannot bet more than you have. Say that a gambler ``succeeds'' if, along the infinite stream of bets, her wealth tends to infinity. The ``home team'', consisting of the casino owner and gambler 0, wins if and only if gambler 0 succeeds while the other gamblers do not.
\par
In the case depicted above, gambler 0 uses even integers as wagers, and the others use odd wagers. We will see that the home team can guarantee a win in this case and that they could not had it been the other way round, namely, odd integers for gambler 0 and even integers for the others. In general, let $A$ be the wagers allowed to gambler 0, and $B$ those allowed to the other gamblers ($A,B$ being any pair of subsets of the positive real numbers). We ask for which $A$ and $B$ the home team can \emph{guarantee a win}, in the following sense: gambler 0 can choose a strategy that guarantees that for any strategies that the other gamblers choose the casino can choose a sequence so that the home team wins.
\par
As a simple example, suppose that apart from gambler 0 there is only one more, regular, gambler (call him gambler 1), and first, let $A=\{2\}$ and $B=\{1\}$. Then the home team cannot guarantee a win: for any pure strategy of gambler 0, let gambler 1 employ the strategy that is a copy of it, only wagering 1 instead of 2. Then the gains (or losses) of 1 are exactly half the gains of 0. Therefore, if gambler 0 succeeds, so does gambler 1. Second, let $B$ be as before, and let $A=\{1,2\}$. Then the home team can guarantee a win as follows. The strategy of gambler 0 is simple: she always bets on red; at the start she wagers two dollars, and keeps doing this as long as she wins. Then, after the first loss, she switches to wagering one dollar ad infinitum. The casino strategy chooses, in a first phase, red every time, until the wealth of 0 exceeds 1's wealth by more than three dollars (note that this is bound to eventually happen, since in the first phase 0 gains two every time, while 1 gains at most one). Then it chooses black once (after which the wealth of 0 still exceeds the wealth of 1, as 0 lost two dollars and 1 gained no more than one dollar). Then comes the second phase, where the casino always chooses the opposite color to what gambler 1 bet on, and chooses red if 1 did not bet anything. This strategy of the casino ensures that, during the second phase, 1 cannot make a non-zero bet more than a finite number of times before he goes broke (and if 1 does go broke, 0 still does not). Hence, whatever strategy he chooses, 1 will not succeed, while 0 will.
\subsection{Computability, randomness and unpredictability}
\par
This casino game is related to the notions of randomness and computability. In the theory of \emph{algorithmic randomness}, a sequence of zeros and ones is called \emph{computably random} if there exists no computable strategy that succeeds against it (thus, the sequence is in some sense unpredictable). For a set $A$ of positive real numbers, a sequence is \emph{$A$-valued random} if there exists no computable gambling strategy (aka ``martingale''), with wagers in $A$, that succeeds against it. Hence we may compare two sets of reals $A$ and $B$, and ask whether $B$-valued randomness implies $A$-valued randomness. In other words, whether for any casino sequence $x$, and a computable $A$-martingale (gambling strategy with wagers in $A$) that succeeds on $x$, there exists a $B$-martingale that succeeds on $x$. And in terms of our casino game, suppose that \emph{every} computable $B$-martingale is employed by some regular gambler. Then the question is whether there exist a computable $A$-martingale and a casino sequence, such that the home team wins.
\par
There are countably many computable strategies. In many cases this countability, rather than the specific type of admissible strategies, is the essential point in the analysis. This naturally leads to the following setting (Definition \ref{def countable anticipation}): gambler 0's wagers are in $A$, and there are countably many regular players, whose wagers are in $B$.
We will say that $A$ \emph{evades} $B$ if the home team can guarantee a win. This is the central notion in our paper. Thus, we are not directly concerned with computability, or other complexity considerations\footnote{Rod Downey (private communication) pointed out that in the case where $A$ and $B$ are recursive sets of rational numbers our analysis applies to the question whether every $B$-valued random is $A$-valued     random.}, but rather take a more abstract view: we only require that the number of strategies against which the home team needs to concurrently combat is countable (See also Remark \ref{rem:effective} below). One may consider a few variants of this settings (see below).
\par
If $A$ does not evade $B$, we say that $B$ \emph{anticipates} $A$. If $B$ contains $A$, then clearly $B$ anticipates $A$, because the strategy of gambler 1 can be an exact copy of 0's strategy (and there is even no need for countably many regular gamblers, one is enough). Therefore, if 0 succeeds, so does 1. Similarly, if $B$ contains a multiple of $A$, i.e., $B \supseteq  r A$ for some $r > 0$, then $B$ anticipates $A$: as in the example above, the $B$-strategy can be the same as the $A$-strategy up to the multiple $r$ (i.e., the wagers are multiplied by $r$), and therefore the gains are the same, multiplied by $r$. For example, the even integers anticipate the odd integers (or the whole set of integers, for that matter).
\par
Thus, $B$ containing a multiple of $A$ is a sufficient condition for $B$ to anticipate $A$. Is it also a necessary condition? \cite{chalcraft12} showed that if $A$ and $B$ are finite, then it is necessary. They asked whether this characterization extends to infinite $A$ and $B$ (note that their framework is that of algorithmic randomness, i.e., it only allows for computable martingales). In other words:
\begin{itemize}
\item[$(*)$] Does $B$ not containing a multiple of $A$ imply that $A$ evades $B$?
\end{itemize}
\par
\subsection{Our contribution}
Theorem~\ref{thm:suffic} implies a negative answer to this question.\footnote{We restrict our attention to closed sets, since every set anticipates its closure (Lemma~\ref{lem close}).} Thus, for example, the segment $[0,1]$ anticipates the set $\mathbb R_+$, although it does not contain any multiple of it. Still, although $(*)$ does not hold in general, we will see classes of pairs $A,B$ where $(*)$ does hold. Theorem~\ref{theorem bounded} says that if $A$ is bounded, and $B\setminus \{0\}$ is bounded away from 0, then $(*)$ holds. Thus, for example, the set $\{1,\pi\}$ evades the positive integers.
Theorem \ref{theorem well ordered} says that if $B$ is well-ordered, then $(*)$ holds. Thus, for example, the set of even integers evades the odds.
\par
Now let us consider the following version of our casino game, whose description is more like that of a ``classic'' repeated game. At each stage of the repeated game, first gambler 0 makes a bet; then, after observing it, the (team of) regular gamblers each make their bets; and then the casino chooses a red/black outcome. All of our results apply to this version as well.
\par
The difference between our definition of evasion and the description of this game is that in the former setting, the regular gamblers knew the \emph{strategy} of gambler 0 (i.e., including her future actions), gambler 0 did not observe their past actions, and the casino knew the strategies of all gamblers. One may consider many such variations to the settings, and our results are robust to all of them -- whether gambler 0 observes anything or not, whether the regular gamblers know her future bets or not, and whether the casino knows the future bets or not.
\par

Previous work on martingales with restricted wagers (\cite{bienvenu12}, \cite{teutsch13}, \cite{peretz13}) employed various solutions in various specific situations. The present paper proposes a systematic treatment that applies to most of the previously studied situations, as well as many others. Our proofs are elementary and constructive. Also, the proposed construction is recursive: it does not rely on what the strategies do in the future; therefore we believe our method can be useful in many frameworks where computational or other constraints are present.
\par
The rest of the paper is organized as follows. Section \ref{sec:Martingales} presents the definitions, results, and a few examples, as well as a discussion of related previous work. The next two sections contain proofs. Section \ref{sec:Discussion} points to open problems and further directions.
\section{Definitions and Results}\label{sec:Martingales}
A \emph{martingale} is a gambling strategy that bets on bits of a binary sequence.
Formally, it is a function $M:\{\head,\tail\}^{*}\to\mathbb R$ that satisfies
\[
M(\sigma)=\frac{M(\sigma\head)+M(\sigma\tail)}{2},
\]
for every string $\sigma\in\{\head,\tail\}^{*}$.

The \emph{increment} of $M$ at $\sigma\in\{\head,\tail\}^{*}$ is defined as
\[
M'(\sigma)= M(\sigma\head)-M(\sigma).
\]

For $A\subset\mathbb R_+$, we say that $M$ is an \emph{$A$-martingale} if $|M'(\sigma)|\in A$, for every $\sigma\in\{\head,\tail\}^{*}$.

The empty string is denoted $\varepsilon$ and $M(\varepsilon)$ is called the
\emph{initial value} of $M$. Note that a martingale is determined by its initial
value and its increments.

The initial sub-string of length $n$ of a binary sequence, $X\in\{\head,\tail\}^{\infty}$, is denoted $X\restriction n$. A martingale $M$ \emph{succeeds} on $X$, if $\lim_{n\to\infty}M(X\restriction n)=\infty$ and $M(X\restriction n)\geq |M'(X\restriction n)|$, for every $n$. The latter condition asserts that $M$ doesn't bet on money it doesn't have. The set of sequences on which $M$ succeeds is denoted $\Succ(M)$. A martingale $N$ \emph{dominates} $M$ if $\Succ(N)\supseteq\Succ(M)$, and a set of martingales $\mathcal N$ dominates $M$ if $\bigcup_{N\in\mathcal N}\Succ(N)\supseteq\Succ(M)$.

The following are non-standard definitions.
\begin{definition}\label{def single anticipation}
A set $B\subseteq\mathbb R_+$ \emph{singly anticipates} a set $A\subseteq\mathbb R_+$, if every $A$-martingale is dominated by some $B$-martingale. If $A$ singly anticipates $B$ and $B$ singly anticipates $A$, we say that $A$ and $B$ are \emph{strongly equivalent}.
\end{definition}

\begin{definition}\label{def countable anticipation}
A set $B\subseteq\mathbb R_+$ \emph{countably anticipates} (\emph{anticipates}, for short) a set $A\subseteq\mathbb R_+$, if every $A$-martingale is dominated by a countable set of $B$-martingales. If $A$ anticipates $B$ and $B$ anticipates $A$, we say that $A$ and $B$ are \emph{(weakly) equivalent}. If $B$ does not anticipate $A$ we say that $A$ \emph{evades} $B$.
\end{definition}

Note that both ``singly anticipates'' and ``anticipates'' are reflexive and transitive relations (namely, they are preorders),\footnote{The evasion relation is anti-reflexive and is not transitive.} and that single anticipation implies anticipation. Also, if $A\subseteq B$ then $B$ singly anticipates $A$.

\begin{remark}\label{rem:effective}
The motivation for the definition of countable anticipation comes from the study of algorithmic randomness (see, e.g., \cite{downey-online}), where it is natural to consider the countable set of all $B$-martingales that are computable relative to $M$  (see \cite{peretz13}). Formally, a set $B\subseteq\mathbb R_+$ \emph{effectively anticipates} a set $A\subseteq\mathbb R_+$, if for every $A$-martingale $M$ and every sequence $X\in\Succ(M)$, there is a $B$-martingale, computable relative to $M$, that succeeds on $X$. More generally, one could define in the same fashion a ``$\mathcal C$-anticipation'' relation with respect to any complexity class $\mathcal C$. Our main focus will be on countable anticipation, and specifically on sets $A$ and $B$ such that $B$ does not countably anticipate $A$; and therefore $B$ does not $\mathcal C$-anticipate $A$ for any complexity class $\mathcal C$. When we present cases in which anticipation does hold, the dominating martingales will usually be fairly simple relative to the dominated martingales. We do not presume to rigorously address the computational complexity of those reductions, though.
\end{remark}

The topological closure of a set $A\subseteq\mathbb R_+$ is denoted $\bar A$. The following lemma says that we can restrict our attention to closed subsets of $\mathbb R_+$.
\begin{lemma}\label{lem close}
Every subset of $\mathbb R_+$ is strongly equivalent to its closure.
\end{lemma}
\begin{proof}{}
Let $A\subset\mathbb R_+$ and let $M$ be an $\bar A$-martingale. Define
an $A$-martingale $S$ by
\begin{align*}
S(\varepsilon)&=M(\varepsilon)+2,\\
S'(\sigma)&\in A\cap (M'(\sigma)-2^{-|\sigma|},M'(\sigma)+2^{-|\sigma|} ),
\end{align*}
where $|\sigma|$ is the length of $\sigma$. Clearly, $S(\sigma) > M(\sigma)$,
for every $\sigma\in\{\head,\tail\}^*$; therefore $\Succ(S)=\Succ(M)$.
\end{proof}

Another simple observation is that for every $A\subseteq\mathbb
R_+$ and $r>0$, $A$ is strongly equivalent to $rA\:=\{ra: a\in A\}$. This
observation leads to the next definition.
\begin{definition}
Let $A,B\subseteq\mathbb R_+$. We say that $A$ and $B$ are \emph{proportional}, if there exists $r>0$ such that {$r A=B$}.
If we only require that $rA\subseteq \bar B$, then $A$ is proportional to a subset of the closure of $B$. In that case we say that $A$ \emph{scales} into $B$.
\end{definition}

From the above and the fact that $A\subseteq B$ implies that $B$ singly anticipates $A$, we have the following lemma.
\begin{lemma}\label{lem similarity}
If $A$ scales into $B$, then $B$ anticipates $A$, for every $A,B\subseteq\mathbb R_+$.
\end{lemma}
The next two theorems provide conditions under which the converse of Lemma~\ref{lem similarity} also holds.
\begin{theorem}\label{theorem bounded}
For every $A,B\subseteq\mathbb R_+$, if $\sup A<\infty$ and $0\not\in \overline{B\setminus\{0\}}$, then $B$ anticipates $A$ only if $A$ scales into $B$.
\end{theorem}

\begin{theorem}\label{theorem well ordered}
For every $A,B\subseteq\mathbb R_+$, if $B$ is well ordered, namely, $\forall x\in\mathbb R_+\ x\not\in \overline{B\setminus[0,x]}$, then $B$ anticipates $A$ only if $A$ scales into $B$.
\end{theorem}

\citet{chalcraft12} studied effective anticipation between finite subsets of $\R$. They showed that (effective) anticipation  is equivalent to containing a proportional set on the domain of finite sets. They further asked whether their result extends to infinite sets. In particular, they asked if $\mathbb Z_+$ anticipates the set $V=\{0\}\cup[1,\infty)$. \cite{peretz13} showed that the set $\{1+\frac{1}{n}\}_{n=1}^\infty\subset V$ evades (i.e., is not anticipated by) $\mathbb Z_+$, and the set $\{\frac{1}{n}\}_{n=1}^\infty$ evades $V$. All of the above results follow immediately from Theorem \ref{theorem bounded}. Furthermore, any set that contains two $\mathbb Q$-linearly independent numbers (e.g., $\{1,\pi\}$) evades $\mathbb Z_+$.

In light of Lemma \ref{lem close}, one may rephrase the above question of \citet{chalcraft12} and ask whether, for infinite sets, anticipation is equivalent to scaling. The answer is still negative (Theorem \ref{thm:suffic}), although Theorems \ref{theorem bounded} and \ref{theorem well ordered} above gave some conditions under which the equivalence does hold.

Theorem \ref{theorem well ordered} says that if, for example, $A = \mathbb Z_+$ and $B$ is a subset of $\mathbb Z_+$ whose density is zero, then $A$ evades $B$. This is because $B$ is well ordered, and $A$ does not scale into $B$ (by the zero density). Another example is when $B$ is the set of all odd integers where, again, $A$ evades $B$. Note that the set of even integers is proportional to $\mathbb Z_+$ and hence is equivalent to $\mathbb Z_+$. Therefore the even integers evade the odd integers, but not vice versa.

Furthermore, we can take any subset of $\mathbb Z_+$ that does not contain an ``ideal'' (i.e., all the integer multiples of some number, namely, a set proportional to $\mathbb Z_+$). If $B$ is of the form $B = \mathbb Z_+ \setminus \event{n \cdot \phi(n)}_{n=1}^\infty$ for some function $\phi : \mathbb N \to \mathbb N$, then $\mathbb Z_+$ evades $B$ even when, for example, the function $\phi$ grows very rapidly. In particular, the density of $B$ could equal one.

The previous theorems gave some necessary conditions for anticipation. The next theorem gives a sufficient condition.
\begin{theorem}\label{thm:suffic}
Let $A,B\subset\R_+$\,.
For $x>0$, let
$$
P(x)=
\{ t \geq 0 : t\cdot (A\cap [0,x])  \subseteq \bar{B}\cup \event{0} \,\}.
$$
For any $M \geq 0$ \,let $q_M (x) =
\max (P(x) \cap [0,M])$.
If for some $M$, $\int_0^\infty q_M(x)dx = \infty$, then $B$ singly anticipates $A$.
\end{theorem}
This is equivalent to the following seemingly stronger theorem.
\newtheorem*{nt}{Theorem \ref{thm:suffic}$^*$}
\begin{nt}\label{lemma f}
Let $A,B\subset\R_+$. Suppose there is a non-increasing function $f\colon\R_+\to\R_+$, such that
\begin{enumerate}
\item $\bar{B}\supset f(x)\left(A\cap[0,x]\right)$, for every $x\in\R_+$; and
\item $\int_0^\infty f(x)\,\mathrm d x=\infty$.
\end{enumerate}
Then, $B$ singly anticipates $A$.
\end{nt}
\begin{proof}[Proof of Theorem \ref{thm:suffic}$^*$]
Let $M$ be an $A$-martingale. By Lemma~\ref{lem close} we may assume that $B$ is closed. Define a $B$-martingale, $S$, by
\begin{figure}
\centering
\includegraphics[width=\textwidth,keepaspectratio=true,trim=100pt 103pt 32pt 79pt,clip=true]{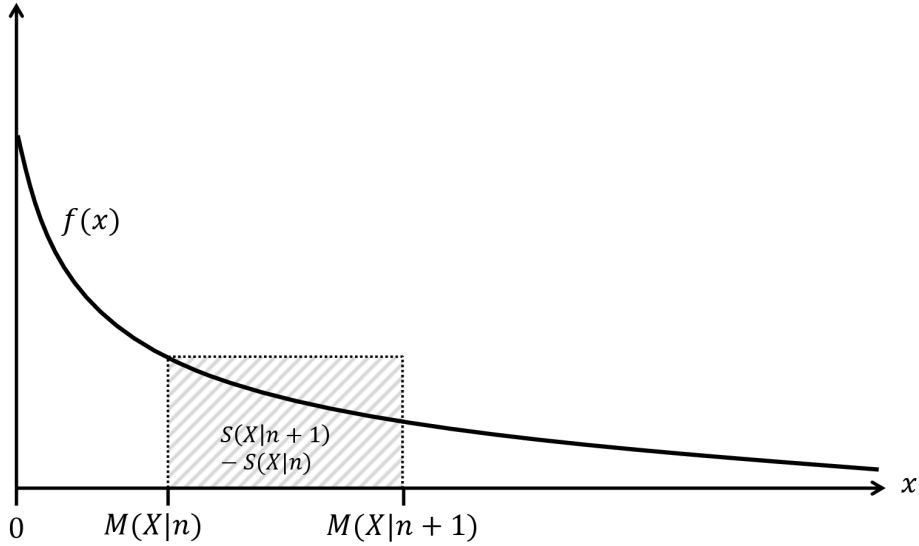} 
\caption{Inequality \eqref{eq integral}}
\end{figure}
\begin{align*}
S(\varepsilon) &= f(0)M(\varepsilon),\\
S'(\sigma)&=f(M(\sigma))M'(\sigma).
\end{align*}
Let $X\in\{\head,\tail\}^{\mathbb N}$ such that $X \in \Succ(M)$, and let $n\in\mathbb N$.
Since $f$ is non-increasing,
\begin{equation}\label{eq integral}
S(X\restriction n+1)-S(X\restriction n)\geq\int_{M(X\restriction n)}^{M(X\restriction n+1)}f(x)\,\mathrm{d}x.
\end{equation}
It follows by induction that
\[
S(X\restriction n)\geq\int_{0}^{M(X\restriction n)}f(x)\,\mathrm{d}x,
\]
which concludes the proof, since $M(X\restriction n) \to \infty$ and $\int_{0}^{\infty}f(x)\,\mathrm{d}x=\infty$.
\end{proof}
If, for example, $A$ scales into $B$, namely, $r A \subseteq \bar{B}$, the function $f$ in Theorem \ref{thm:suffic}$^*$ can be taken to be simply $f(x)=r$. Also note that when $A$ and $B$ are finite, such a function $f$ as in the theorem exists if and only if $A$ scales into $B$.

The theorem tells us, for example, that although $\R_+$ does not scale into the interval $[0,1]$, these two sets are (strongly) equivalent: to see that $[0,1]$ singly anticipates $\mathbb R_+$, apply Theorem \ref{thm:suffic}$^*$ with $f(x)=\min\{\frac 1 x  ,1 \}$.

Another example is the set $A = \{2^n\}_{n=-\infty}^{+\infty}$ being (strongly) equivalent to  $B = \{2^n\}_{n=-\infty}^0$, although $A$ does not scale into $B$. To see this, apply Theorem \ref{thm:suffic}$^*$ with $f(x)=\min\{1/  2^{\lfloor\log_2 x\rfloor} ,1 \}$.

We previously saw, by Theorem \ref{theorem well ordered}, that $A = \mathbb Z_+ $ evades $B = \mathbb Z_+ \setminus \event{n \cdot \phi(n)}_{n=1}^\infty$. Now look at this example in the context of Theorem \ref{thm:suffic}, to see that there is no contradiction. Suppose $\phi(n)$ is increasing, and $\phi(n) \to \infty$. Then for any $x > 0 $, the set $P(x)$ in the theorem is unbounded, i.e., $\sup P(x)=\infty$. Nevertheless, fix any choice of $M$, and note that for every $x$ large enough, $P(x)$ does not contain any nonzero number smaller than $M$, i.e., $P(x) \cap [0,M] = \event{0}$. Thus, for any choice of $M$, $q_M(x) = 0$ for every $x$ large enough. In particular, $\int_0^\infty q_M(x)dx < \infty$.
\section{Proof of Theorem~\ref{theorem bounded}}
Throughout this section $A,B\subset\mathbb R_+$ are two sets satisfying
\begin{align*}
\sup A &<\infty,\\
0&\not\in \overline {B\setminus\{0\}},\\
A&\text{ does not scale into } B.
\end{align*}
We must show $A$ (countably) evades $B$.

Since $A$ does not scale into $B$, one thing that $B$-martingales cannot do in general is to mimic $A$-martingales, not even up to a constant ratio. We use this idea in order to construct a sequence of heads and tails that will separate between the two types of martingales.

\subsection{Ratio minimization}
Let $N$ and $M$ be martingales with $N$ non-negative. We say that $x\in\{\head,\tail\}$ is the \emph{$N/M$-ratio-minimizing outcome} at $\sigma\in\{\head,\tail\}^{<\infty}$ (assuming $M(\sigma)> 0$) if either
\begin{enumerate}
\item $\frac{N(\sigma)}{M(\sigma)}> \frac{N(\sigma x)}{M(\sigma x)}$, or
\item $\frac{N(\sigma)}{M(\sigma)}= \frac{N(\sigma x)}{M(\sigma x)}$ and $M(\sigma x)> M(\sigma)$, or
\item $M'(\sigma)=N'(\sigma)= 0$ and $x=\head$.
\end{enumerate}
In words: our first priority is to make the ratio $N/M$ decrease; if this is impossible (i.e., the increments of $N$ and $M$ are proportional to their value at $\sigma$, and so $N/M$ doesn't change), then we want $M$ to increase so as to insure that $M(\sigma x)>0$; if that is impossible as well (i.e., both increments are 0), we set $x$ to be $\head$, for completeness of definition only.

Our definition extends to finite/infinite extensions of $\sigma$ by saying that $X$ is the length $|X|$ (possibly $|X|=\infty$) \emph{$N/M$-ratio-minimizing extension} of $\sigma$, if $X_{t+1}$ is the $N/M$-ratio-minimizing outcome at $X\restriction t$, for every $|\sigma|\leq t < |X|$.

For any such $N$, $M$, and $\sigma$ the infinite $N/M$-ratio-minimizing extension of $\sigma$, $X$, makes the ratio $N(X\restriction t)/M(X\restriction t)$ monotonically converging to a limit $L\in \mathbb R_+$, as $t\to\infty$. The next lemma will help us argue that the ratio between the increments ${N'(X\restriction t)}/M'(X\restriction t)$ also converges to $L$ in a certain sense.
\begin{lemma}[Discrete l'H\^opital rule]\label{lem:lhopital}
Let $(a_n)_{n=1}^\infty$ and $(b_n)_{n=1}^\infty$ be sequences of real numbers. Assume that $a_n> 0$, for every $n$. If $b_n/{a_n}$ monotonically converges to a limit $L\in\mathbb R$, and $\sup\{\frac 1 n \sum_{k=1}^n |a_{k+1}-a_k|\}<\infty$, then
\begin{equation*}
\lim_{n\to\infty} \frac 1 n \sum_{k=1}^{n} |(b_{k+1}-b_k) - L(a_{k+1}-a_k)|=0.
\end{equation*}
\end{lemma}
\begin{proof}
Since $\frac 1 n \sum_{k=1}^n |a_{k+1}-a_k|$ is bounded, $\frac 1 n \sum_{k=1}^n (a_{k+1}-a_k)$ is bounded, too; therefore
\begin{equation}\label{eq weak}
\lim_{n\to\infty} \frac 1 n \sum_{k=1}^{n} (b_{k+1}-b_k) - L(a_{k+1}-a_k)=0.
\end{equation}
It remains to prove that
\begin{equation}\label{eq strong'}
\lim_{n\to\infty} \frac 1 n \sum_{k=1}^{n} \left[b'_k-L a'_k\right]_+=0,
\end{equation}
where $a'_k=a_{k+1}-a_k$ and similarly $b'_k=b_{k+1}-b_k$.

We may assume w.l.o.g. that $\frac{b_n}{a_n}\searrow L$ (otherwise consider the sequence $(-b_n)_{n=1}^\infty$). Namely, $\frac{b_k}{a_k}\geq \frac{b_{k+1}}{a_{k+1}}$, which implies that $b'_k\leq \frac{b_k}{a_k}a'_k$; hence
\[
\left[b'_k-L a'_k\right]_+ \leq \left[\left(\frac{b_k}{a_k}-L\right)a'_k\right]_+\leq \left(\frac{b_k}{a_k}-L\right)|a'_k|.
\]
Now \eqref{eq strong'} follows since $\frac{b_k}{a_k}$ converges to $L$ and $\frac 1 n \sum_{k=1}^n |a'_k|$ is bounded.
\end{proof}

\begin{cor}\label{cor ratio}
Let $M$ and $N$ be a pair of martingales and $\sigma\in \{\head,\tail\}^{*}$. Assume that $N$ is non-negative, $M(\sigma)>0$, and $M'$ is bounded.\footnote{The assumption that $M'$ is bounded can be relaxed by assuming only that $\frac{1}{N}\sum_{t=0}^{N-1}|M'(X\restriction t)|$ is bounded.} Let $X$ be the infinite $N/M$-ratio-minimizing extension of $\sigma$ and $L=\lim\limits_{t\to\infty}\frac{N(X\restriction t)}{M(X\restriction t)}$. For every $\epsilon>0$ the set
\[
\left\{t:\left\vert N'(X\restriction t) -L\cdot{M'(X\restriction t)} \right\vert > \epsilon \right\}
\]
has zero density.
\end{cor}
\begin{proof}
Note that
\begin{multline*}
\left\vert N'(X\restriction t) -L\cdot M'(X\restriction t)\right\vert =\\
\left\vert (N(X\restriction t+1)-N(X\restriction t)) -L\cdot(N(X\restriction t+1)-N(X\restriction t))\right\vert.
\end{multline*}
By Lemma~\ref{lem:lhopital}, we have
\[
\lim_{n\to\infty} \frac 1 n \sum_{k=1}^{n} |N'(X\restriction t) -L\cdot M'(X\restriction t)|=0,
\]
which implies that the density of the set
\[
\left\{t:\left\vert N'(X\restriction t) -L\cdot{M'(X\restriction t)}\right\vert > \epsilon \right\}
\]
is zero, for every $\epsilon>0$.
\end{proof}

The next step is to construct the history-independent $A$-martingale prescribed by Theorem~\ref{theorem bounded}. We formalize the properties of this martingale in the following lemma.
\begin{lemma}\label{lemma 0 ratio}
Let $A,B\subset\mathbb R_+$. Suppose that $\sup A< \infty$ and $A$ does not scale into $B$; then there exists a history-independent $A$-martingale with positive initial value, $M$, such that for every non-negative $B$-martingale, $N$, and every $\sigma\in\{\head,\tail\}^*$ such that $M(\sigma)>0$, the infinite $N/M$-ratio-minimizing extension of $\sigma$, $X$, satisfies
\[
\lim_{t\to\infty}\frac{N(X\restriction t)}{M(X\restriction t)}=0.
\]
\end{lemma}
Note that the $A$-martingale, $M$, does not depend on $N$ or $\sigma$, so the same $M$ can be used against any $N$ at any $\sigma$ that leaves $M(\sigma)$ positive.

\begin{proof}[Proof of Lemma \ref{lemma 0 ratio}]
Let $\{a_n\}_{n=0}^\infty$ be a countable dense subset of $A\setminus\{0\}$. For every positive integer $t$, let $n(t)\in\mathbb Z_+$ be the largest integer such that $2^{n(t)}$ divides $t$. Define $x_t\:=a_{n(t)}$. The sequence $\{x_t\}_{t=1}^\infty$ has the property that the set
\begin{equation}\label{eq positive density}
\{t:|x_t-a|<\epsilon\}
\end{equation}
has positive density, for every $\epsilon>0$ and $a\in A$.

Let $M$ be a history-independent martingale whose increment at time $t$ is $x_t$, for every $t\in\mathbb Z_+$ (with an arbitrary positive initial value). Let $N$ be an arbitrary non-negative $B$-martingale. Suppose that $M(\sigma)> 0$ and let $X$ be the infinite $N/M$-ratio-minimizing extension of $\sigma$ and let $L=\lim_{t\to\infty}\frac{N(X\restriction t)}{M(\restriction t)}$. Corollary \ref{cor ratio} and \eqref{eq positive density} guarantee that $L\cdot A\subset \bar B$. By the assumption that $A$ does not scale into $B$, we conclude that $L=0$.
\end{proof}
\subsection{The casino sequence}
In the rest of this section we assume that $\sup A=\inf (B\setminus \{0\})=1$. This is w.l.o.g. since proportional sets are (strongly) equivalent.

We begin with an informal description of the casino sequence. Lemma \ref{lemma 0 ratio} provides a history-independent $A$-martingale, $M$, that can be used against any $B$-martingale.

Given a sequence of $B$-martingales, $N_1,N_2,\ldots$, we start off by ratio-minimizing against $N_1$. When $M$ becomes greater than $N_1$, we proceed to the next stage. We want to make sure that $N_1$ no longer makes any gains. This is done by playing adversarial to $N_1$ whenever he wagers a positive amount.

At times when $N_1$ wagers nothing (i.e., $N_1'=0$), we are free to choose either $\head$ or $\tail$ without risking our primary goal. At those times we turn to ratio-minimizing against $N_2$, while always considering the goal of keeping $N_1$ from making gains a higher priority. Since $\inf \left\{B\setminus\{0\}\right\}=1>0$, it is guaranteed that at some point we will no longer need to concern $N_1$, and hence, at some even further point, $M$ will become greater than $N_1+N_2$.

The process continues recursively, where at each stage our highest priority is to prevent $N_1$ from making gains, then $N_2$, $N_3$, and so on until $N_k$; and  if none of $N_1,\ldots,N_k$ wagers any positive amount, we ratio-minimize against $N_{k+1}$.

When a positive wager of some $N_i$, $i\in\{1,\ldots k\}$, is answered with an adversarial outcome, a new index $k'$ must be calculated, so that $M$ is sufficient to keep $N_1,\ldots,N_{k'}$ from making gains. That is, $M>N_1+\cdots+N_{k'}$.

An inductive argument shows that for every fixed $k$, there is a point in time beyond which none of $N_1,\ldots,N_k$ will ever wager a positive amount; therefore, at some even further point, $M$ becomes greater than $N_1+\cdots+N_{k+1}$; hence the inductive step.

The above explains how the value of each $N_i$ converges to some $L_i \in \mathbb R_+$, and the limit inferior of the value of $M$ is at least $\sum_{i=1}^{\infty}L_i$. In order to make sure that $M$ goes to infinity we include, among the $N_i$s, infinitely many martingales of constant value $1$.

We turn now to a formal description. As mentioned above, we assume without loss of generality that $\sup A=\inf \{B\setminus \{0\}\}=1$. We additionally assume that $0\in B$, and so we can convert arbitrary $B$-martingales to non-negative ones by making them stop betting at the moment they go bankrupt.

Let $M$ be a history-independent $A$-martingale provided by Lemma \ref{lemma 0 ratio}. Let $N_1,N_2,\ldots$ be a sequence of non-negative $B$-martingales. Assume without loss of generality that infinitely many of the $N_i$s are the constant $1$ martingale.

We define a sequence $X\in\{\head,\tail\}^\infty$ recursively. Assume $X\restriction t$ is already defined.

First we introduce some notation. Denote the value of $M$ at time $t$ by $m(t)=M(X\restriction t)$, and similarly  $n_i(t)=N_i(X\restriction t)$, for every $i\in\mathbb N$. Let
\begin{align*}
S_i(t)&=\sum_{j=1}^in_j(t),\\
k(t)&=\max\{i:S_i(t) < m(t)\},\text{ and}\\
S(t)&=S_{  k(t)  }(t).
\end{align*}
Note that the maximum is well defined, since the $n_i(t)$ include infinitely many 1s.

We are now ready to define $X_{t+1}$. We distinguish between two cases: Case I: there exists $1\leq j\leq k(t)$ such that $N_j'(X\restriction t)\neq 0$; Case II: $N_1'(X\restriction t)=\cdots=N_{k(t)}'(X\restriction t)= 0$.

In Case I, let $i=\min\{j:N_j'(X\restriction t)\neq 0\}$ and define
\[
X_{t+1} =
\begin{cases}
  \tail &\text {if $N_{i}'(X\restriction t)>0$,}\\
  \head &\text {if $N_{i}'(X\restriction t)<0$.}
\end{cases}
\]

In Case II, $X_{t+1}$ is the $\frac{N_{k(t)+1}}{M-S(t)}$-ratio-minimizing outcome at $X\restriction t$. Explicitly,
\[
X_{t+1} =
\begin{cases}
  \tail &\text {if $\frac{N_{k(t)+1}'(X\restriction t)}{M'(X\restriction t)}>\frac{n_{k(t)+1}}{m(t)-S(t)}$,}\\
  \head &\text {if $\frac{N_{k(t)+1}'(X\restriction t)}{M'(X\restriction t)}\leq\frac{n_{k(t)+1}}{m(t)-S(t)}$.}
\end{cases}
\]

Consider the tuple $\alpha(t)=(\lfloor n_1(t)\rfloor,\ldots,\lfloor n_{k(t)}(t)\rfloor)$. In Case I, $\alpha(t+1)$ is strictly less that $\alpha(t)$ according to the lexicographic order. In Case II, $\alpha(t)$ is a prefix of $\alpha(t+1)$, and so under a convention in which a prefix of a tuple is greater than that tuple, we have that $\{\alpha(t)\}_{t=1}^\infty$ is a non-increasing sequence.\footnote{Alternatively, one can use the standard lexicographic order where $\alpha(t)$ is appended with an infinite sequence of $\infty$ elements.} Let $k=\liminf\limits_{t\to\infty} k(t)$. It follows that from some point in time, $\alpha$ consists of at least $k$ elements; therefore the first $k$ elements of $\alpha$ must stabilize at some further point in time. Namely, for $t$ large enough we have $\lfloor n_i(t)\rfloor=\lim\limits_{t'\to\infty}\lfloor n_i(t')\rfloor <\infty$, for every $i\leq k$. Since the increments of $n_i(t)$ are bounded below by 1, $n_1(t),\ldots,n_k(t)$ stabilize, too. Also, since $m(t)> S(t)$, we have $\liminf\limits_{t\to\infty}m(t)\geq \lim_{t\to\infty} S_i(t)$, for every $i\leq \liminf\limits_{t\to\infty} k(t)$. Since there are infinitely many $i$'s for which $n_i(t)$ is constantly $1$, the proof of Theorem \ref{theorem bounded} is concluded by showing that $\liminf\limits_{t\to\infty} k(t)=\infty$.

Assume by negation that $\liminf\limits_{t\to\infty} k(t)=k<\infty$. There is a time $T_0$ such that $n_k(t)=\lim\limits_{t'\to\infty}n_k(t')$ and $k(t)\geq k$, for every $t>T_0$. There cannot be a $t>T_0$, for which $k(t)>k$ and $k(t+1)=k$. That would mean a Case I transition from time $t$ to $t+1$ and we would have $n_i(t+1)<n_i(t)$, for some $i \leq k$. It follows that $k(t)=k$, for every $t>T_0$.

From time $T_0$ ratio-minimization against $n_{k+1}$ takes place. Let $l=\linebreak\liminf\limits_{t\to\infty}n_{k+1}(t)$. If $l=0$, then $n_{k+1}(t)<1$, for some $t>T_0$; at this point $n_{k+1}$ stabilizes (otherwise  $N_{k+1}$ would go bankrupt); therefore $n_{k+1}(t)=0$; therefore $k(t)>k$, which is not possible. If $l>0$, then by Lemma~\ref{lemma 0 ratio}, there must be some time $t>T_0$ in which $m(t)>n_{k+1}(t)+S_k(T_0)=n_{k+1}(t)+S(t)$, which contradicts the definition of $S(t)$.

\section{Proof of Theorem~\ref{theorem well ordered}}
To show that if $B$ is well ordered and $A$ does not scale into $B$, then $A$ evades $B$, we construct an $A$-martingale $M$, s.t. for any $B$-martingales $N_1,N_2,\ldots$ we construct a sequence $X$ on which $M$ succeeds, while every $N_i$ does not.

We begin with a rough outline of the proof ideas. $M$ always bets on ``heads.'' Before tackling every $N_i$, we first gain some money and ``put it  aside.'' Then we ratio-minimize against $N_i$. It will eventually make $M$ sufficiently richer than $N_i$, so that we can declare $N_i$ to be ``fragile'' now. This means that from now on the casino can make $N_i$ lose whenever it is ``active'' (i.e., makes a non-zero bet), since $M$ can afford losses until $N_i$ is bankrupt. When $N_i$ is not active, we can start tackling $N_{i+1}$, while constantly making sure that we have enough money kept aside for containing the fragile opponents. An important point is to show that once some $N_i$ becomes fragile, it remains fragile unless a lower-index martingale becomes active.
\par
Let $(a_n)$ be a sequence that is dense within $A\setminus \event{0}$, and such that each number in the sequence appears infinitely many times. For example, given a dense sequence $(x_n)$ in $A\setminus \event{0}$, the sequence
\[
x_1;\ x_1,x_2;\ x_1,x_2,x_3;\ x_1,x_2,x_3,x_4; \ldots
\] can be used.
\par
Since multiplying $A$ or $B$ by a positive constant does not make a difference, we may assume w.l.o.g. that $a_1=1$, and that $\inf (B\setminus \event{0}) =1$ ($B$ is well ordered; hence in particular $B\setminus \event{0}$ is bounded away from $0$).
\par
We construct $M$ and $X$ as follows. Denote $m(t) = M(X\restriction t)$, and similarly $m'(t) = M'(X\restriction t)$. Take integers
$$f(t,k) \geq \max \event{m(t),\abs{(a_{k+1}-a_k)/a_k}}.$$
We take $M(\ve) = a_1$. For the first $f(0,1)$ stages, $M' = a_1$. Then, after stage $t = f(0,1)$, $M' = a_2$ for the next $f(t,2)$ periods, and after stage $t' = t + f(t,2)$, $M' = a_3$ for $f(t',3)$ stages, and so on. But this goes on only as long as no ``tails'' appears. Whenever a ``tails'' appears, namely, at a stage $t$ where $X_t = \tail$, we revert to playing from the beginning of the sequence, i.e., $M' = a_1$ for the next $f(t,1)$ stages (or until another ``tails'' appears), then $M' = a_2$, etc.
\par
For $t \geq 0$ we define the function $\g(t)$, which is similar to $m'(t)$, but modifies sudden increases of $m'$ into more gradual ones.
\begin{figure}
\centering
\includegraphics[width=\textwidth,keepaspectratio=true,trim=56pt 165pt 90pt 198pt,clip=true]{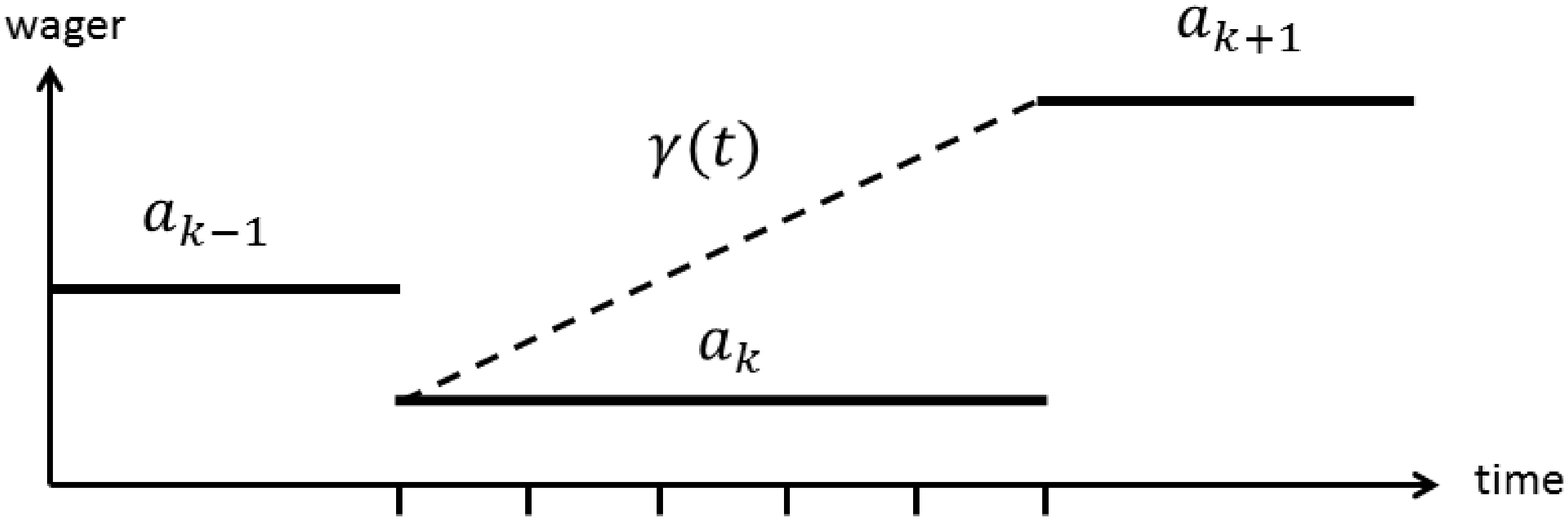} 
\caption{$\g$ and the wagers}
\end{figure}
At the beginning of a block of stages where $a_k$ is wagered (i.e., where $m'(\cdot)=a_k$), $\g$ equals $a_k$. If $a_{k+1} \leq a_k$, then $\g$ remains $a_k$ throughout this block. Otherwise, it linearly increases until reaching $a_{k+1}$ exactly at the beginning of the next block. I.e., suppose $m'(t) = a_k$, and let $s \leq t$ be the beginning of the block (of length $f(s,k)$) of $a_k$ wagers. If $a_{k+1} \leq a_k$ then $\g(t) = a_k$. Otherwise,
$$\g(t) = \frac{(s+ f(s,k) - t) \,a_k + (t-s) a_{k+1}}{f(s,k)}.$$
%
\begin{align*}
\text{Note:\qquad}  &(i)\,\g(t) \geq m'(t),\\
&(ii) \text{\,If } X_{t+1} = \tail  \text{ \,then } \g(t+1) = a_1 = 1,\\
&(iii) \text{\,If } X_{t+1} = \head \text{ \,then } \g(t+1)- \g(t) \leq m'(t).
\end{align*}
The last one follows from $m'(t) = a_k$, \,$\g(t+1)- \g(t) \leq (a_{k+1} - a_k) / f(s,k)$, and $f(s,k) \geq \abs{(a_{k+1}-a_k)/a_k}$\,, by the definition of $f$.
\par
Let $N_1,N_2,\ldots$ be $B$-martingales (and assume these martingales never bet on money that they do not have). Denote $n_i(t) = N_i(X\restriction t)$. To define the sequence $X$, denote $\nu_k(t) = k + n_1(t)+\ldots+n_k(t)$, and let $p=p(t) \geq 0$ be the largest integer such that
$$
m(t) - (\g(t) - 1) >   \nu_p(t).
$$
$N_1,\ldots,N_{p(t)}$ are the ``fragile'' martingales at time $t$. Define
$$\mu(t) = m(t)- (\nu_p(t) + 1)$$
and consider two cases. (i) If there exists some index $1 \leq j \leq p(t)$ s.t. $n'_j(t) \neq 0$, let $i$ be the smallest such index, and $X_{t+1}$ is chosen adversely to $n'_i$. (ii) Otherwise, $X_{t+1}$ is chosen by $\mu / N_{p+1}$-ratio-minimizing, i.e., if $\mu(t) > 0 $ and $n'_{p+1}(t) / m'(t) > n_{p+1}(t) / \mu(t)$ then $X_{t+1}= \tail$, and otherwise $X_{t+1}= \head$\,.
\par
We now show that these $M$ and $X$ indeed work.
\begin{lemma}\label{lem:fragile}
For any $t$, if $p(t) \geq i $ and $n'_j(t)=0$ for every $j < i$, then $p(t+1) \geq i$\,.
\end{lemma}
\def \lside{\mathcal{L}}
\begin{proof}[Proof of Lemma \ref{lem:fragile}]
We prove the following equivalent claim:\\(I) If $i \leq p(t)$ is the smallest index such that $n'_i(t) \neq 0$, then $i \leq p(t+1)$. (II) If  $n'_j(t) = 0$ for any $j \leq p(t)$, then $p(t) \leq p(t+1)$.
\par
In case (II), denote $i=p(t)$. Then in both cases $$m(t) - (\g(t) - 1) >   \nu_i(t)$$ is known. Let $\lside(t)$ designate the LHS of this inequality. We need to show that $\lside(t+1) > \nu_i(t+1)$. Note that for any $j< i$, $n_j(t+1) = n_j(t)$, and that if $X_{t+1} = \head$ \,then $\lside(t+1) \geq \lside(t)$, since $m(t+1) = m(t)+ m'(t)$ and $\g(t+1) \leq \g(t) + m'(t)$.
\par
(I) In this case the casino makes $i$ lose, hence $n_i(t+1) \leq n_i(t)-1$ (recall that $\inf (B\setminus \event{0})=1$); therefore $\nu_i(t+1) \leq \nu_i(t)-1$. If $X_{t+1} = \head$ \,we are done. If $X_{t+1} = \tail$ \,then $m(t+1) = m(t) - m'(t)$, and $\g(t+1)=1$; therefore, $\lside(t+1) = \lside(t) - m'(t) + (\g(t) -1 ) \geq \lside(t) - \g(t) + (\g(t) -1 ) = \lside(t) -1$.
\par
(II) In this case $n_i(t+1) = n_i(t)$; hence $\nu_i(t+1) = \nu_i(t)$. If $X_{t+1} = \head$ \,we are done. If $X_{t+1} = \tail$ \,then $m(t+1) = m(t) -m'(t)$ and $\g(t+1) = 1$. But $X_{t+1} = \tail$ \,also implies (by the definition of $X$) that $\mu(t) > 0$ and $n'_{i+1}(t) / m'(t) > n_{i+1}(t) / \mu(t)$. Since $n'_{i+1}(t)$ is always $\leq n_{i+1}(t)$, we get that $\mu(t) > m'(t)$. Now, $\lside(t+1) = (m(t) - m'(t)) - (1-1 ) = m(t) - m'(t) > m(t) - \mu(t) = \nu_i(t)+1$, because $\mu(t)$ is \,$m(t) - (\nu_i(t) +1)$. Thus, $\lside(t+1) > \nu_i(t)+1 > \nu_i(t) = \nu_i(t+1)$.
\end{proof}
Remark: The above argument also proves that $M$ is never bankrupt, i.e., $m(t) \geq m'(t)$, and moreover $m(t) \geq \g(t)$, as follows.
\par
In the beginning $1=m(0)\geq \g(0)=1$. As long as $p(t)=0$
we are in case (II). In this case, if $X_{t+1}=\head$ \,then $\lside(t)$ does not decrease; hence $m(t)-\g(t)$ does not decrease. And if $X_{t+1}=\tail$ \,we just saw that actually $\lside(t+1) > 1 + \nu_i(t+1)$, which implies that $m(t)-\g(t) > 0$.
\par
Once $p(t)>0$, then $\nu_i(t) \geq 1$, hence $\lside(t) > \nu_i(t)$ implies that $m(t)-\g(t) > 0$; and Lemma \ref{lem:fragile} implies that $p(t)$ remains $> 0$.
\begin{lemma}\label{lem:inductive}
For any $i$ there exists a stage $T_i$ \,s.t. for any $t > T_i$\,, $n'_1(t) = n'_2(t) = \ldots = n'_i(t) = 0$, and $p(t) \geq i$.
\end{lemma}
\begin{proof}[Proof of Lemma \ref{lem:inductive}]
We proceed by induction over $i \geq 0$; namely, the induction hypothesis is that the lemma holds for $i-1$. Note that the induction base case \,$i=0$ holds vacuously.\footnote{Incidentally, the inequality defining fragility always holds for $p(t)=0$, as $m(t) \geq \g(t)$ and $\nu_0(t) = 0$ imply that $m(t) - (\g(t)-1) > \nu_0(t)$.}
\par
If $p(t_0) \geq i$ for some stage $t_0 > T_{i-1}$, then $p(t) \geq i$ for every $t \geq t_0$, by lemma \ref{lem:fragile}. From this stage on, the casino chooses adversely to $i$ whenever $i$ is active (because the lower-index players are not active). Therefore, $i$ will be active at no more than $n_i(t_0)$ stages after $t_0$, since  afterwards $i$ has nothing to wager, and we are done.
\par
So assume by way of contradiction that $p(t) = i-1$ for every $t > T_{i-1}$\,. Then $X_{t+1}$ is $\mu/N_i$-ratio-minimizing. As long as $\mu(t) \leq 0$ we get ``heads''; therefore, from some stage on, $\mu > 0$ (as every $a_k$ appears infinitely many times, the sum of the wagers will not converge). Denote $q(t) = n_i(t)/ \mu(t)$. $q(t) \geq 0$ is non-increasing and therefore converges to a limit $L$. Denote $0 \leq r(t) = n_i(t) - L  \mu(t) $.

Suppose there exists some $k$ s.t. the wagers $m'(t)$ never reach beyond $a_1,\ldots,a_k$\,. Hence, there are infinitely many stages $t$ where $i$ over-bets, i.e., $n'_i(t) > q(t)m'(t)$. For $1 \leq j \leq k$, let $x_j = L  \,a_j$. Since $B$ is well ordered, there exists some $\d_j > 0$ \,s.t. $(x_j,\,x_j+\d_j) \cap B = \emptyset$\,. Since $q(t) \to L$\,, $q(t) < L + \min_{1\leq j \leq k}(\d_j/a_j)$ \,for $t$ large enough.
\par
When $i$ over-bets and $m'(t) = a_j$, then $r(t+1) = n_i(t+1) - L \mu(t+1) \leq n_i(t) - (La_j + \d_j) - L (\mu(t) - a_j )=  r(t) - \d_j$. When $i$ does not over-bet, then $n'_i(t) \leq x_j = L a_j$\,, and $r(t+1) = n_i(t+1) - L \mu(t+1) \leq (n_i(t)+La_j) - L(\mu(t) + a_j) =   r(t) $. Therefore $r(t)$ does not increase, and infinitely many times it decreases by at least $\d = \min_{1 \leq j\leq k}\d_j > 0$; hence eventually $r(t) < 0$, which is a contradiction.
\par
Therefore, there does not exist an index $k$ as above. This implies that for any $j$, there is a stage $t$ after which $a_j$ is wagered $f(t,j)$ consecutive times, and $N_i$ does not over-bet (otherwise $a_{j+1}$ cannot be reached). Now suppose that $L>0$\,. Let $A_0 = \event{a_1,a_2,\ldots}$ be the set of all the values that the sequence $(a_n)$ takes. $A_0$ is dense in $A$, and $L\cdot A \nsubseteq \bar B$ (since $A$ does not scale into $B$); therefore, also $L\cdot A_0 \nsubseteq \bar B$. Hence, there exists an $a \in A_0$ s.t. the distance between $La$ and $\bar B$ is $\d > 0$.
\par
Let $\D = \min \event{\d/a, \d}$. For $t$ larger than some $T_\D$, \,$q(t) < L + \D$. Since $a$ appears infinitely many times in the sequence $(a_n)$, there exist $T_a >T_\D$ and an index $j$, s.t. $a_j = a$ is wagered $f(T_a,j)$ consecutive times, and at each of these times $n'_i(t) \leq La-\d$, as otherwise $n'_i(t) \geq La + \d$, but that is over-betting since $(La + \d )/ a  = L + \d/a > q(t)$.
Hence, $r(t+1) \leq n_i(t)+La-\d - L(\,\mu(t) + a) = r(t) - \d$. But $q(T_a) < L + \d$; therefore $n_i(T_a) < (L+\d)\,\mu(T_a)$; hence $r(T_a) < \d \mu(T_a)$. By the definition of $f$, $f(T_a,j) \geq m(T_a) \geq \mu(T_a)$; therefore after those $f(T_a,j)$ times, $r < 0$\,. This cannot be; therefore $L=0$.
\par
As $q(t) \to 0$, surely $q(t) < 1$ for large enough $t$, namely, $\mu(t) > n_i(t)$. Since $\mu(t) = m(t) - (\nu_{i-1}(t)+1)$, we get $m(t) > n_i(t) + \nu_{i-1}(t) + 1 = n_i(t) + ((i-1) + n_1(t) + \ldots + n_{i-1}(t)) + 1 = i + n_1(t) + \ldots + n_i(t) = \nu_i(t)$. At some stage $t$, $M$ starts wagering $1$.
For this $t$, $\g(t) = m'(t) =1$; hence $m(t) - (\g(t)-1) = m(t) > \nu_i(t)$, contradicting our assumption that $i$ is not fragile.
\end{proof}
Lemma \ref{lem:inductive} states that any $N_i$ is only active a finite number of times, and therefore it is bounded; it also states that for any $i$ and large enough $t$, $m(t) - (\g(t)-1) >  \nu_i(t)$, hence $m(t) >  \nu_i(t) + (\g(t)-1) >  \nu_i(t) -1 \geq \,i -1$, and therefore $m(t) \to \infty $.
\section{Extensions and Further Research}\label{sec:Discussion}
\subsection{Extensions}
The definition of anticipation corresponds to the game where gambler 0 first announces her $A$-martingale, then the regular gamblers announce their $B$-martingales, and then the casino chooses a sequence. Thus, the regular gamblers know the future actions of gambler 0, the casino knows the future actions of all gamblers, and gambler 0 does not even observe the past bets of other gamblers.
\par
Our proofs, however, do not rely on this state of affairs. One may consider variants of this game in which the casino does not know the future actions, or the regular gamblers do not know the future actions of gambler 0, or gambler 0 does observe the bets of others; or any combination thereof. Our results hold for all these variants.
\par
One of these variants looks more like a ``classic'' repeated game: at each stage of the game, first gambler 0 makes a bet, then the regular gamblers make bets, and then the casino chooses red or black (and everything is observed by all).
\par
Another generalization is that the set $B$ need not be the same for all regular gamblers. That is, the results still apply for sets $B_1,B_2,B_3\ldots$ of wagers for gamblers $1,2,3\ldots$, if the conditions hold for each $B_i$ separately (of course, the conditions may not hold for their union).
\par
\cite{bienvenu12} consider the case of computable martingales, with $A$ being the set of all integers, and $B_i$ finite sets of integers. They use a probabilistic argument to show evasion. Our proof provides, in particular, a construction of the casino sequence for this case. The constructed sequence is computable relative to the history of bets.

\subsection{Further Research}
It seems possible that our proof of Theorem~\ref{theorem well ordered} could be modified so as to avoid the assumption that $B$ is well ordered. We conjecture that Theorems~\ref{theorem bounded} and \ref{theorem well ordered} can be unified as the following statement.
\begin{conjecture}
Let $A,B\subset \mathbb R_+$. If $0\not\in\overline {B \setminus \{0\}}$, then $B$ anticipates $A$ only if $A$ scales into $B$.
\end{conjecture}
The present paper strove to understand the effect of restricting the wager sets on the prediction power of martingales. As is often the case, understanding one thing brings up many new questions. We list just a few.
\begin{itemize}
\item Under the assumptions of Theorem~\ref{theorem bounded}, $A$ can evade $B$ through a history-independent martingale. Is this also the case under the assumptions of Theorem~\ref{theorem well ordered}?
\item Are single anticipation and countable anticipation different? That is, are there sets $A,B\subset\mathbb R_+$, such that $B$ countably, but not singly, anticipates $A$?
\item What can be said about the anticipation relation between sets that \emph{do} include 0 as an accumulation point, for example, $\{2^{-n}\}_{n=1}^\infty$, $\{\frac 1 n\}_{n=1}^\infty$, and $\mathbb R_+$?
\item \cite{buss} introduced martingales defined by probabilistic strategies. How do these martingales behave in our framework?
\end{itemize}

\section{Acknowledgement}
We wish to thank Abraham Neyman and Bernhard von Stengle for their useful comments.


\begin{thebibliography}{}
\bibitem[Bienvenu et~al., 2012]{bienvenu12}
Bienvenu, L., Stephan, F., and Teutsch, J. (2012).
\newblock How Powerful Are Integer-Valued Martingales?
\newblock {\em Theory of Computing Systems}, 51(3):330--351.

\bibitem[Buss and Minnes, 2013]{buss}
Buss, S. and Minnes, M. (2013).
\newblock Probabilistic Algorithmic Randomness.
\newblock {\em Journal of Symbolic Logic}, 78(2):579--601.

\bibitem[Chalcraft et~al., 2012]{chalcraft12}
Chalcraft, A., Dougherty, R., Freiling, C., and Teutsch, J. (2012).
\newblock How to Build a Probability-Free Casino.
\newblock {\em Information and Computation}, 211:160--164.

\bibitem[Downey and Riemann, 2007]{downey-online}
Downey, R. G. and Riemann, J. (2007).
\newblock Algorithmic Randomness.
\newblock {\em Scholarpedia} 2(10):2574,
\newblock http://www.scholarpedia.org/article/algorithmic\_\linebreak{}randomness.

\bibitem[Peretz, 2013]{peretz13}
Peretz, R. (2013).
\newblock Effective Martingales with Restricted Wagers.
\newblock http://arxiv.org/abs/1301.7465.

\bibitem[Teutsch, 2013]{teutsch13}
Teutsch, J. (2013).
\newblock A Savings Paradox for Integer-Valued Gambling Strategies.
\newblock {\em International Journal of Game Theory}, forthcoming.

\end{thebibliography}

\end{document}